\documentclass[
final
, nomarks
]{dmtcs-episciences}

\usepackage[utf8]{inputenc}
\usepackage{subfigure}
\usepackage{enumerate}

\usepackage{amsmath,amssymb,amsfonts,amstext}
\usepackage{bbm}
\usepackage{enumitem,mdwlist} 
\usepackage{cases} 
\usepackage{varwidth} 
\usepackage{subfigure} 
\usepackage{mathrsfs} 
\usepackage{complexity}

\newtheorem{theorem}{Theorem}[section]
\newtheorem{problem}[theorem]{Problem}

\newtheorem{corollary}[theorem]{Corollary}

\newtheorem{proposition}[theorem]{Proposition}

\usepackage[round]{natbib}

\author{Diane Castonguay
  \and Erika M. M. Coelho
  \and Hebert Coelho
  \and Julliano R. Nascimento}
\title[A note on the convexity number of complementary prisms]{A note on the convexity number of complementary prisms\thanks{This study was financed in part by the Coordenação de Aperfeiçoamento de Pessoal de Nível Superior - Brasil (CAPES) - Finance Code 001}}
\affiliation{
  Instituto de Inform\'{a}tica, Universidade Federal de Goi\'{a}s, GO, Brazil}

\received{2018-9-25}

\revised{2019-7-5}

\accepted{2019-7-5}

\begin{document}

\publicationdetails{21}{2019}{4}{4}{4853}

\maketitle
\begin{abstract}
  In the geodetic convexity, a set of vertices $S$ of a graph $G$ is \textit{convex} if all vertices belonging to any shortest path between two vertices of $S$ lie in $S$. The cardinality $con(G)$ of a maximum proper convex set $S$ of $G$ is the \textit{convexity number} of $G$. The \textit{complementary prism} $G\overline{G}$ of a graph $G$ arises from the disjoint union of the graphs $G$ and $\overline{G}$ by adding the edges of a perfect matching between the corresponding vertices of $G$ and $\overline{G}$. In this work, we prove that the decision problem related to the convexity number is $\NP$-complete even restricted to complementary prisms, we determine $con(G\overline{G})$ when $G$ is disconnected or $G$ is a cograph, and we present a lower bound when $diam(G) \neq 3$. 
\end{abstract}

\section{Introduction} 
\label{sec:introduction}

In this paper we consider finite, simple, and undirected graphs, and we use standard terminology. For a finite and simple graph $G$ with vertex set $V(G)$, a \textit{graph convexity} on $V(G)$ is a collection $\mathcal{C}$ of subsets of $V(G)$ such that $\emptyset, V(G) \in \mathcal{C}$ and $\mathcal{C}$ is closed under intersections.
The sets in $\mathcal{C}$ are called \textit{convex sets} and the \textit{convex hull} $H(S)$ in $\mathcal{C}$ of a set $S$ of vertices of $G$ is the smallest set in $\mathcal{C}$ containing $S$.

Convex sets in graphs emerged as an analogy to convex sets in the Euclidean plane. Such concepts have attracted attention in the last decades due to its versatility for modeling some disseminating processes between discrete entities.
Graph convexities can model, for instance, contexts of distributed computing \citep{flocchini2004dynamic,peleg2002local}. 


We can consider a computer network, modeled as a graph $G$ (the computers and their connections are represented by $V(G)$ and $E(G)$, respectively), where a fault on some computer data propagates to other computers, according to some rule of propagation.
A rule of propagation may be, for instance, if a fault occurs in two computers $a$ and $b$, such a fault propagates to all computers that lie in the shortest paths between $a$ and $b$. A problem raised in this context is findinding a maximum subset of computers $S \subseteq V(G)$ in which faults can occur in order to ensure that the entire network does not fail.
The cardinality of $S$ corresponds to the parameter known as convexity number of the graph $G$ and the propagation rule coincides with the geodetic convexity.

We may cite other contexts in which graph convexities can be applied, \textit{e.g.} spread of disease and contamination \citep{balogh1998random,bollobas2006art,dreyer2009irreversible}, marketing strategies \citep{domingos2001mining,kempe2003maximizing,kempe2005influential}, and spread of opinion \citep{brunetti2012minimum,dreyer2009irreversible}.

 We consider the \textit{geodetic convexity} $\mathcal{C}$ on a graph $G$, which is defined by means of shortest paths in $G$. We say that a set of vertices $S$ of a graph $G$ is \textit{convex} if all vertices belonging to any shortest path between two vertices of $S$ lie in $S$. The cardinality $con(G)$ of a maximum proper convex set $S$ of $G$ is the \textit{convexity number} of $G$. 

One of the first works to introduce the convexity number was published by \cite{chartrand2002convexity}. They determine the convexity number for complete graphs, paths, cycles, trees, and present bounds for general graphs.
In the same year, \cite{canoy2002convex} show results on the convexity number for graph operations like join, composition, and Cartesian product. Later on, \cite{kim2004lower} studies the parameter for $k$-regular graphs. Considering complexity aspects,  \cite{gimbel2003some} shows that determining the convexity number is $\NP$-hard for general graphs, whereas \cite{dourado2012convexity} refine  Gimbel's  result showing the $\NP$-hardness of the problem even restricted to bipartite graphs.

Motivated by the work of \cite{canoy2002convex} on the convexity number for graph operations, we study that parameter for a graph product called complementary prism. Such graph product was introduced by \cite{haynes2007complementary} as a variation of the well-known \textit{prism} of a graph \citep{hammack2011handbook}. Let $G$ be a graph and $\overline{G}$ its complement. For every vertex $v \in V(G)$ we denote $\overline{v} \in V(\overline{G})$ as its \textit{corresponding vertex}. The \textit{complementary prism} $G\overline{G}$ of a graph $G$ arises from the disjoint union of the graph $G$ and $\overline{G}$ by adding the edges of a perfect matching between the corresponding vertices of $G$ and $\overline{G}$. A classic example of a complementary prism is the Petersen graph $C_5\overline{C}_5$.

In this paper we determine the convexity number for complementary prisms $G\overline{G}$ when $G$ is disconnected or $G$ is a cograph, and we present a lower bound of that parameter for complementary prisms of graphs with restricted diameter. Furthermore, given a complementary prism $H\overline{H}$, and an integer $k$, we prove that it is $\NP$-complete to decide whether $con(H\overline{H}) \geq k$.  

This paper is divided in more three sections. In Section~\ref{sec:preliminaries} we define the fundamental concepts and terminology. In Section~\ref{sec:results} we present our contributions. We close with the conclusions in Section~\ref{sec:conclusions}.

\section{Preliminaries}
\label{sec:preliminaries}
 

Let $G$ be a graph. Given a vertex $v \in V(G)$, its \textit{open neighborhood} is denoted by $N_G(v)$, and its \textit{closed neighborhood}, denoted by $N_G[v]$, is the set $N_G[v] = N_G(v) \cup \{v\}$.  For a set $U \subseteq V(G)$, let $N_G(U) = \bigcup_{v \in U} N_G(v) \setminus U$, and $N_G[U] = N_G(U) \cup U$. 

A \textit{clique} (resp. \textit{independent set}) is a set of pairwise adjacent (resp. nonadjacent) vertices.
A vertex of a graph $G$ is \textit{simplicial} in $G$ if its neighborhood induces a clique.

The \textit{distance} $d_G(u,v)$ of two vertices $u$ and $v$ in $G$ is the minimum number of edges of a path in $G$ between $u$ and $v$. Let $A, B \subseteq V(G)$. The \textit{distance} between $A$ and $B$ in $G$ is defined by $d_G(A,B) = \min\{d_G(u,v): x \in A, \, y \in B\}$.
A graph $G$ is called \textit{connected} if any two of its vertices are linked by a path in $G$. Otherwise, $G$ is called \textit{disconnected}. A maximal connected subgraph of $G$ is called a \textit{connected component} or \textit{component} of~$G$. A component $G_i$ of a graph $G$ is \textit{trivial} if  $|V(G_i)| = 1$, and \textit{non-trivial} otherwise.

Let $G$ be a graph. For a set $X \subseteq V(G)$, we let $\overline{X}$ be the \textit{corresponding set} of vertices in $V(\overline{G})$.
We denote the set of positive integers $\{1, \dots, k\}$ by $[k]$.

A convex set $S$ of a graph $G$ can be defined by a closed interval operation. 
The \textit{closed interval} $I[u,v]$ of a pair $u,v \in V(G)$ consists of all vertices lying in any shortest $(u,v)$-path in $G$. For a set $S \subseteq V(G)$, the \textit{closed interval} $I[S]$ is the union of all sets $I[u,v]$ for $u,v \in S$ and if $|S| < 2$, then $I[S] = S$. We say that $S$ is a \textit{convex set}, if $I[S] = S$. To avoid ambiguity, sometimes a subscript can be added to the notation (\textit{e.g.} $I_G[S]$, and $H_G(S)$) to indicate which graph $G$ is being considered.

\section{Results} 
\label{sec:results}

\cite{chartrand2002convexity} provide two useful results. They proved that $con(K_n) = n-1$, and for a noncomplete graph the result follows in Theorem~\ref{theo:chartrand2002}.

\begin{theorem}[\cite{chartrand2002convexity}]\label{theo:chartrand2002}
Let $G$ be a noncomplete connected graph of order $n$. Then $con(G) = n-1$ if and only if $G$ contains a simplicial vertex.
\end{theorem}

We begin our contributions by determining in Theorem~\ref{theo:convexityDisconnected} the convexity number for complementary prisms of disconnected graphs. 
We first show Proposition~\ref{prop:kInPath} that will be useful for the subsequent results.
 
\begin{proposition} \label{prop:kInPath}
Let $G$ be a graph, $S \subseteq V(G\overline{G})$, and $v_1 \dots v_k$ be a path in $G$, for $k \geq 2$. 
If $\{v_1,\overline{v}_2,\dots,\overline{v}_k\}$ $\subseteq H(S)$, then $v_k \in H(S)$.
\end{proposition}

\begin{proof}
The proof is by induction on $k$. First, let $k = 2$. Since $v_1v_2 \in E(G)$ and $v_1, \overline{v}_2 \in H(S)$, $v_2 \in I[v_1, \overline{v}_2]$. Now, let $k > 2$. Let $v_1 \dots v_{k-1}v_k$ be a path in $G$ and suppose that $\{v_1,\overline{v}_2,\dots,\overline{v}_{k-1},\overline{v}_k\} \subseteq H(S)$. By induction hypothesis $v_{k-1} \in H(S)$, which implies that $v_k \in I[v_{k-1}, \overline{v}_k]$. Therefore, it follows that $v_k \in H(S)$, for $k \geq 2$. 
\end{proof}

\begin{theorem}
\label{theo:convexityDisconnected}
Let $G$ be a disconnected graph of order $n$, and $k$ be the order of a minimum component of $G$.  Then, $con(G\overline{G}) = 2n - k$.
\end{theorem}

\begin{proof}
Let $G_1, \dots, G_\ell$ be the components of $G$, for $\ell \geq 2$. We can sort the components $G_1, \dots, G_\ell$ of $G$ in non-decreasing order of size (number of vertices). Then, $G_1$ is a component of minimum order, say $|V(G_1)| = k$. 
If $k = 1$, then $G_1$ is a trivial component. Hence, the unique vertex $v \in V(G_1)$ is a simplicial vertex, and the result $con(G\overline{G}) = n-1$ follows from Theorem~\ref{theo:chartrand2002}. 
We then consider that $|V(G_i)| \geq 2$, for every $i \in [\ell]$.
Let $S = V(G\overline{G}) \setminus V(G_1)$. See Figure~\ref{fig:con_desc} for an illustration.
We show that $S$ is a convex set of $G\overline{G}$.

Suppose, by contradiction, that $S$ is not convex. Then, there must be a shortest path of the form $\bar{x}xy\bar{y}$ for $x,y \in V(G_1)$. But this is a contradiction, since $d_G(\bar{x},\bar{y}) = 2$.

%

Next, we show that $S$ is maximum. By contradiction, suppose that there exists a convex set $S' \subseteq V(G\overline{G})$ such that $|S'| > |S|$. 

Since $|S'| > 2n - k$, we have that $S' \cap V(G_j) \neq \emptyset$, and $S' \cap V(\overline{G}_j) \neq \emptyset$, for every $j \in [\ell]$. We divide the proof in two cases.

\medskip
\noindent {\bf Case 1} {\it $S' \cap V(\overline{G})$ is not a clique.}
\medskip

Let $\overline{x}, \overline{y} \in S' \cap V(\overline{G})$ such that $ \bar{x} \bar{y} \notin E(\overline{G})$. 
By the definition of complementary prism, every vertex in $V(\overline{G}_i)$ is adjacent to every vertex in $V(\overline{G}_j)$, for every $i,j \in [\ell]$, $i \neq j$. 
This implies that $\overline{x},\overline{y} \in V(\overline{G}_i)$, for some $i \in [\ell]$, and $V(\overline{G}) \setminus V(\overline{G}_i) \subseteq I_{G\overline{G}}[\overline{x},\overline{y}]$. 

We know that $|V(G_k)| \geq 2$, for every $k \in [\ell]$. Let $j \in [\ell] \setminus \{i\}$. Since $G_j$ is a connected component, and $|V(G_j)| \geq 2$, there exist $u,v \in V(G_j)$ such that $uv \in E(G)$. Thus, $\bar{u}\bar{v} \notin E(\overline{G})$. Since $\bar{u}, \bar{v} \in I_{G\overline{G}}[\overline{x},\overline{y}]$, we obtain that $V(\overline{G}_i) \subseteq I_{G\overline{G}}[\overline{u},\overline{v}]$. Hence, $V(\overline{G}) \subseteq  H_{G\overline{G}}(\{\overline{x}, \overline{y}\})$.
Since $S'$ is convex, then $V(\overline{G}) = S' \cap V(\overline{G})$. Since $S' \cap V(G_i) \neq \emptyset$, for every $i \in [\ell]$, and $G_i$ is connected, Proposition~\ref{prop:kInPath} implies that $V(G) \subseteq H_{G\overline{G}}(S')$, a contradiction.

\medskip
\noindent {\bf Case 2} {\it $S' \cap V(\overline{G})$ is a clique.}
\medskip

Consider that $S' \cap V(\overline{G})$ is maximum clique. 
Let $\overline{C}_i = S' \cap V(\overline{G}_i)$, for every $i \in [\ell]$. We know that $C_i$ is an independent set, and $\overline{C}_i$ is a clique. We claim that $S' \cap (N_G(C_i) \cup \overline{N_G(C_i)})  = \emptyset$, for every $i \in [\ell]$. In fact, we prove a stronger statement, $S' \cap (V(G_i\overline{G}_i) \setminus (C_i \cup \overline{C}_i))  = \emptyset$, for every $i \in [\ell]$.

\medskip
\noindent {\bf Claim 1} {\it For every $i \in [\ell]$, $S' \cap (V(G_i\overline{G}_i) \setminus (C_i \cup \overline{C}_i))  = \emptyset$.}
\medskip

\noindent {\it Proof of Claim 1} 
First, recall that $\overline{u}, \overline{v} \in I_{G\overline{G}}[u,v]$, for every $u \in V(G_i)$ and $v \in  V(G_{j})$, for every $i,j \in [\ell]$, $i \neq j$.
Let $i \in [\ell]$. By contradiction, suppose that there exists $v \in S' \cap (V(G_i\overline{G}_i) \setminus (C_i \cup \overline{C}_i))$. 

Since $\overline{C}_i$ is maximum, if $v \in S' \cap (V(\overline{G}_i) \setminus \overline{C}_i )$ then $S' \cap V(\overline{G})$ contains two nonadjacent vertices, a contradiction. Then, suppose that $v \in S' \cap (V(G_i) \setminus C_i )$.
Since $S' \cap V(G_j) \neq \emptyset$, for every $j \in [\ell]$, we have that $\overline{v} \in I_{G\overline{G}}[u,v]$, for some $u \in S' \cap V(G_j)$, $i \neq j$. Since $\overline{C}_i$ is maximum, there exists $\overline{w} \in \overline{C}_i$ such that $\bar{v}\bar{w} \notin E(\overline{G})$. Consequently $S' \cap V(\overline{G})$ contains two nonadjacent vertices, a contradiction. \hfill\ensuremath{\square}

\medskip
\noindent {\bf Claim 2} {\it For every $i \in [\ell]$, $|C_i \cap S'| \leq |V(G_i) \setminus C_i|$.}
\medskip

\noindent {\it Proof of Claim 2} 
Let $i \in [\ell]$. 
By Claim 1 $S' \cap (V(G_i) \setminus C_i) = \emptyset$.
Since $S' \cap V(G_i) \neq \emptyset$, we have that $|C_i \cap S'| \geq 1$. 
By contradiction, suppose that $|C_i \cap S'| > |V(G_i) \setminus C_i|$. 
Let $|C_i \cap S'| = p$. 

If $p = 1$, then $|V(G_i) \setminus C_i| = 0$. Since $C_i$ is an independent set, it follows that $G_i$ is disconnected, a contradiction. Then, consider $p \geq 2$.
Let $C_i  \cap S' = \{u_1, \dots, u_p\}$, and $V(G_i) \setminus C_i = \{ v_1, \dots, v_{p-1}\}$.

Since $C_i$ is an independent set, and $G_i$ is connected, it follows that every vertex in $C_i \cap S'$ is adjacent to at least one vertex in $V(G_i) \setminus C_i$. Since $|C_i \cap S'| > |V(G_i) \setminus C_i|$, we have that there exist $j, j' \in [p]$, $j \neq j'$, such that $u_jv_q, u_{j'}v_q \in E(G)$, for some $q \in [p-1]$. Then, $v_q \in I_{G\overline{G}}[u_j, u_{j'}]$. Since $S'$ is convex, $v_q \in S'$. But that is a contradiction, since $S' \cap (V(G_i) \setminus C_i ) = \emptyset$.
\hfill\ensuremath{\square}
\medskip

To conclude the proof, we show that $|S'| \leq n$. In view of the above statements, for every $i \in [\ell]$, it follows that
\begin{align*}
|S' \cap V(G_i\overline{G}_i)|  &= |S' \cap V(G_i)| + |S' \cap V(\overline{G}_i)| & \\
							    &= |C_i \cap S'| + |\overline{C_i}| & \\
							    &= |C_i \cap S'| + |V(\overline{G}_i)| - |V(\overline{G}_i) \setminus \overline{C}_i| & \\
							    &\leq |V(G_i) \setminus C_i| + |V(\overline{G}_i)| - |V(\overline{G}_i) \setminus \overline{C}_i| & \text{(by Claim 2)} \\
							    &= |V(G_i)|. & 
\end{align*}

Since $|S' \cap V(G_i\overline{G}_i)| \leq |V(G_i)|$, for every $i \in [\ell]$, we obtain that \begin{displaymath} |S'| \leq \sum_{i=1}^{\ell} |V(G_i)| = n\text{,} \end{displaymath} a contradiction. Therefore  $S$ is a maximum convex set of $G\overline{G}$, and $con(G\overline{G}) \leq 2n -k$, which completes the proof. 
\end{proof}

Figure~\ref{fig:con_desc} shows an illustration of a proper convex set of $G\overline{G}$, represented by the black vertices. Consider $G_1$ the component of minimum order of $G$.

\begin{figure}[htb]
\centering
{\setlength{\fboxsep}{10pt}
\setlength{\fboxrule}{0.3pt}
\fbox{
\includegraphics[width=0.32\textwidth]{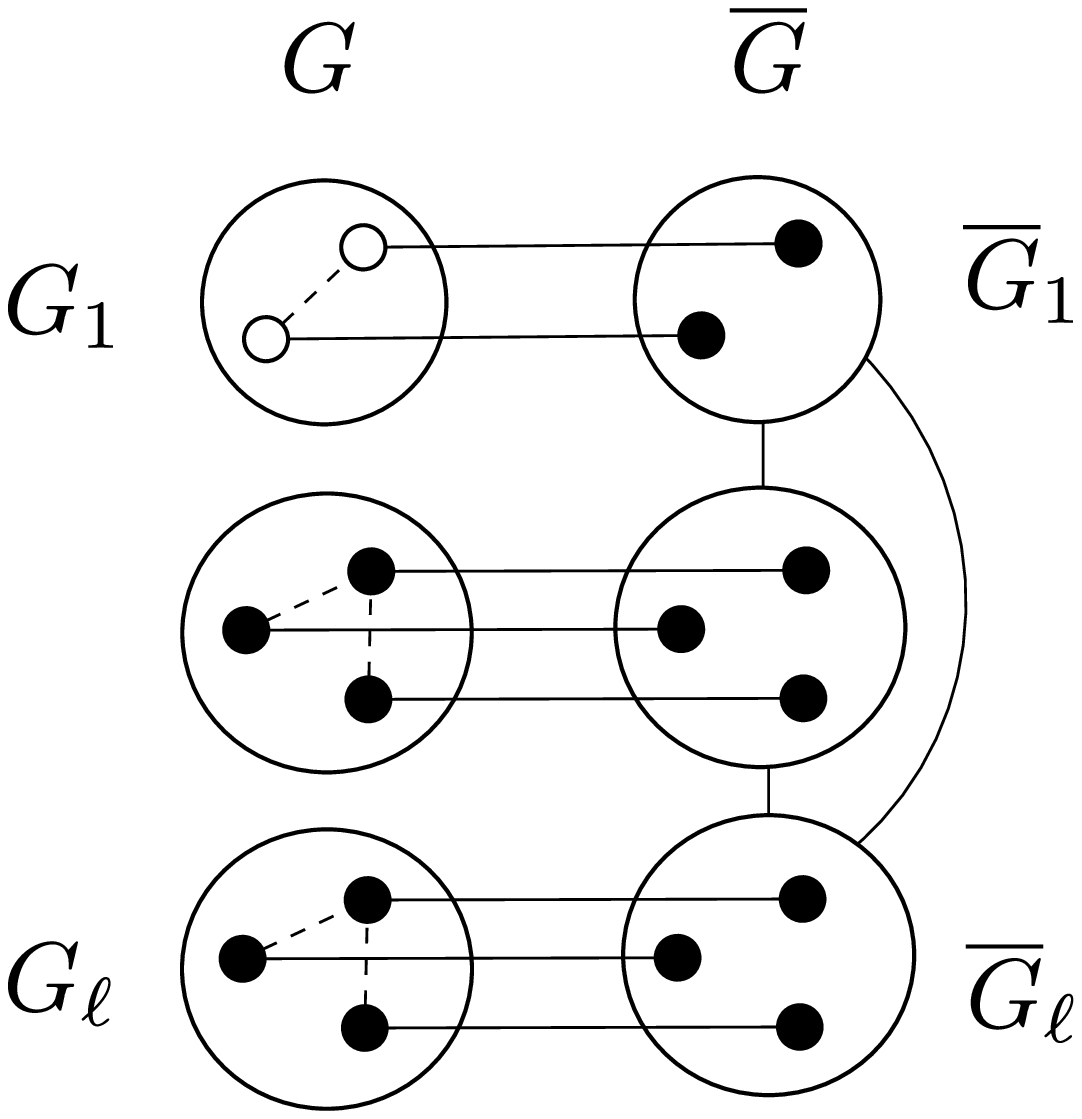}
}}
\caption{Example of a proper convex set of $G\overline{G}$.}
\label{fig:con_desc}
\end{figure}

As a consequence of Theorem~\ref{theo:convexityDisconnected}, we have:

\begin{corollary} \label{cor:convexityCograph}
Let $G$ be a connected cograph of order $n$, and $k$ be the order of a minimum component of $\overline{G}$. Then,
$con(G\overline{G}) = 2n - k$.
\end{corollary}

\begin{proof} 
Since a cograph $G$ is connected if and only if $\overline{G}$ is disconnected~\citep{corneil1981complement}, and $G\overline{G}$ is isomorphic to $\overline{G}G$, the result follows from Theorem~\ref{theo:convexityDisconnected}. 
\end{proof}

\subsection{NP-Completeness}

Next, we show a hardness result of the geodetic number for complementary prims. We first present  Proposition~\ref{prop:vizinhancaConjConvexo}, and we define the two decision problems to be considered.

\begin{proposition}\label{prop:vizinhancaConjConvexo}
Let $G$ be a graph, $S$ be a convex set of $G\overline{G}$, $S' = S \cap V(\overline{G})$. If $\overline{A} = N_{\overline{G}}(S')$ then $S \cap A = \emptyset$.
\end{proposition}

\begin{proof}
Suppose that $u \in S \cap A$. Since $\bar{u}\bar{v} \in E(\overline{G})$, for some $\overline{v} \in S'$, we have that $\overline{u} \in I_{G\overline{G}}[u,\overline{v}]$. We know that $S$ is convex, which implies that $\overline{u} \in S$. But this is a contradiction, since $\overline{u} \notin S \cap V(\overline{G})$. 
\end{proof}

\begin{problem}{\textsc{Clique}} \citep{garey1979computers}.  \\
\textbf{Instance:} A graph $G$ and an integer $k$.\\
\textbf{Question:} Does $G$ have a clique of order at least $k$?
\end{problem}

\begin{problem}{\textsc{Convexity Number}} \citep{dourado2012convexity}.  \\
\textbf{Instance:} A graph $G$ and an integer $k$.\\
\textbf{Question:} Does $G$ have a proper convex set of order at least $k$?
\end{problem}

\begin{theorem}
\label{theo:conNPComplete}
\textsc{Convexity Number} is $\NP$-complete even restricted to complementary prisms.
\end{theorem}

\begin{proof}
Since computing the convex hull of a set of vertices can be done in polynomial time, \textsc{Convexity Number} is in $\NP$.

In order to prove $\NP$-completeness, we describe a polynomial reduction from the $\NP$-complete problem \textsc{Clique} \citep{garey1979computers}. 

Let $(G, k)$ be an instance of \textsc{Clique}. We may assume that $G$ is connected, and $k \geq 3$. Let $|V(G)| = n$. We construct a graph $H$ arising from $G$ as follows. Add to $H$ three cliques $U$, $X$, and $Z$, respectively on $n$, $4n$ and $2$ vertices, say $U = \{u_1, \dots, u_n\}$, $X = \{x_1, \dots, x_{4n}\}$, and $Z = \{z_1, z_2\}$. Add to $H$ an independent set $Y = \{y_1, y_2\}$. Join every vertex in $U \cup Z$ to every vertex in $V(G)$. Also join every vertex in $X$ to $U \cup Y$. This completes the construction of $H$. Use the graph $H$ to create the complementary prism $H\overline{H}$.  See an example in Figure~\ref{fig:construcaoConGrafoVizinhancas}.
We prove that $G$ has a clique of order at least $k$ if and only if $H\overline{H}$ has a proper convex set of order at least $k + 5n + 3$.

First, we assume that $G$ has a clique $C$ of order at least $k$. Let $S \subseteq V(H\overline{H})$ such that $S = C \cup U \cup X \cup Y \cup \{ \overline{u}_1 \}$. Notice that $|S| = k + 5n + 3$. We show that $S$ is a convex set of $H\overline{H}$.

Let $w \in S \cap V(H)$. Since $d_{H\overline{H}}(w,\overline{u}_1) \leq 2$, we have that $I_{H\overline{H}}[w,\overline{u}_1] = \{w,\overline{u}_1 \}$, if $w = u_1$, or $I_{H\overline{H}}[w,\overline{u}_1] = \{w,u_1,\overline{u}_1 \}$, otherwise.
Now, let $w, w' \in S \cap V(H)$. Since $d_H(w,w') \leq 2$, and $C$ is a clique in $G$, we have that 
$I_{H\overline{H}}[w,w'] \subseteq S \cap V(\overline{H})$. Therefore $I_{H\overline{H}}[S] = S$, and $S$ is a convex set of $H\overline{H}$.

For the converse, we show two useful claims first.

\medskip
\noindent {\bf Claim 1} {\it Let $S$ be a proper convex set of $H\overline{H}$. Then, for every $w, w' \in S$, $d_{H\overline{H}}(w,w') \leq 2$.}
\medskip

\noindent {\it Proof of Claim 1} 
By contradiction, suppose that there exist $w, w' \in S$ such that $d_{H\overline{H}}(w,w') > 2$. By the definition of complementary prism, we know that either $w,w' \in V(H)$ or $w,w' \in V(\overline{H})$. 

Let $w,w' \in S \cap V(H)$. By the construction of $H$, $d_H(X \cup Y, Z)=3$; then we can consider that $w \in X \cup Y$ and $w' \in Z$. 
That way, we have that $U \cup V(G) \cup \{\overline{w},\overline{w}'\} \subseteq I_{H\overline{H}}[w,w']$. 
Since $\bar{v}\bar{w} \in E(\overline{H})$ for every $v \in V(G)$, we have that $\overline{v} \in I_{H\overline{H}}[v, \overline{w}]$. Consequently $V(\overline{G}) \subseteq I_{H\overline{H}}[ V(G) \cup \{\overline{w}\}]$. By symmetry,  $\overline{U} \subseteq I_{H\overline{H}}[U \cup \{\overline{w}' \}]$. Since $G$ is connected, there exist two nonadjacent vertices $\overline{v}_1, \overline{v}_2 \in V(\overline{G})$. This implies that $\overline{X} \cup \overline{Y} \subseteq I_{H\overline{H}}[\overline{v}_1, \overline{v}_2]$. Similarly, since $\overline{U}$ is an independent set, $\overline{Z} \subseteq I_{H\overline{H}}[\bar{U}]$. Since $V(\overline{H}) \subseteq H_{H\overline{H}}(S)$, Proposition~\ref{prop:kInPath} implies that  $V(H) \subseteq H_{H\overline{H}}(S)$, a contradiction.

Now, let $\overline{w},\overline{w}' \in S \cap V(\overline{H})$. By the construction of $H$, we can select $\overline{w} \in U$, and $\overline{w}' \in V(\overline{G})$. We have that $\overline{X} \cup \overline{Y} \cup \overline{Z} \cup \{w,w'\} \subseteq I_{H\overline{H}}[\overline{w},\overline{w}']$. 
Consequently, $\overline{U} \subseteq I_{H\overline{H}}[\bar{Z}]$, and $V(\overline{G}) \subseteq I_{H\overline{H}}[\bar{X}]$. Since $V(\overline{H}) \cup \{w,w'\} \subseteq H_{H\overline{H}}(S)$, Proposition~\ref{prop:kInPath} implies that  $V(H) \subseteq H_{H\overline{H}}(S)$, a contradiction. 
\hfill\ensuremath{\square}

\medskip
\noindent {\bf Claim 2} {\it Let $S$ be a proper convex set of $H\overline{H}$. Then $S \cap V(\overline{H})$ is a clique.} 
\medskip

\noindent {\it Proof of Claim 2} 
By contradiction, suppose that there exist $\overline{w}, \overline{w}' \in S \cap V(\overline{H})$ such that $\bar{w} \bar{w}' \notin E(\overline{H})$. Then, we have the following cases.

\medskip
\noindent {\bf Case 1.1} {\it $\overline{w}, \overline{w}' \in \overline{Z}$.}
\medskip

This implies that $\overline{X} \cup \overline{Y} \cup \overline{U} \subseteq I_{H\overline{H}}[\overline{w}, \overline{w}']$. Consequently, $V(\overline{G}) \subseteq I_{H\overline{H}}[\bar{X}]$. Since $d_{H\overline{H}}(\overline{U}, V(\overline{G})) = 3$, by Claim 1, $S$ is not a proper convex set, a contradiction.

\medskip
\noindent {\bf Case 1.2} {\it $\overline{w}, \overline{w}' \in \overline{X} \cup \overline{Y} \cup \overline{U}$.}
\medskip

In this case, $\overline{Z} \in I_{H\overline{H}}[\overline{w}, \overline{w}']$, then the proof follows by Case 1.1.

\medskip
\noindent {\bf Case 1.3} {\it $\overline{w}, \overline{w}' \in V(\overline{G}) \cup \overline{Z}$.}
\medskip

We have that $\overline{X} \in I_{H\overline{H}}[\overline{w}, \overline{w}']$, then the proof follows by Case 1.2.
\hfill\ensuremath{\square}

\medskip

Let $S$ be a proper convex set of $H\overline{H}$ of order at least $k + 5n + 3$.
By Claim 2, we know that $S \cap V(\overline{H})$ is a clique. Let $\overline{C} = S \cap V(\overline{H})$. To proceed with the proof, we show that $\overline{C} \cap (V(\overline{G}) \cup \overline{X} \cup \overline{Y} \cup \overline{Z}) = \emptyset$. For that, we consider two cases, $\overline{C}$ does not contain vertices from $V(\overline{G}) \cup \overline{Z}$, and $\overline{C}$ does not contain vertices from $\overline{X} \cup \overline{Y}$.

\medskip
\noindent {\bf Claim 2.1} {\it $\overline{C} \cap (V(\overline{G}) \cup \overline{Z}) = \emptyset$.}
\medskip

\noindent {\it Proof of Claim 2.1} 
Suppose, by contradiction, that $\overline{C}$ contains a vertex in $V(\overline{G})$ or in $\overline{Z}$.
If $\overline{C} \cap \overline{X} = \emptyset$, then $\overline{X} \subseteq N_{\overline{H}}(\overline{C})$. Otherwise, we have that $\overline{X} \setminus \{\overline{x}_i\} \subseteq N_{\overline{H}}(\overline{C})$, for some $i \in [4n]$. In both cases, we have that $| N_{\overline{H}}(\overline{C}) | \geq 4n-1$. For $\overline{A} = N_{\overline{H}}(\overline{C}) $, Proposition~\ref{prop:vizinhancaConjConvexo} implies that $S \cap A = \emptyset$ then $|S \cap V(H)| = |V(H)| - |A|  \leq 6n+4 - (4n-1)  = 2n+5$.
\hfill\ensuremath{\square}

So far, we conclude that the number of vertices from $S$ in $H$ is at most $2n+5$. It remains to show the maximum number of vertices from $S$ in $\overline{H}$. By the construction of $H$, a clique in $\overline{H}$ of maximum order is a proper subset of $V(\overline{G}) \cup \overline{Y}$, hence $|S \cap V(\overline{H})| < n+2$. Consequently, $|S| = |S \cap V(H)|+|S \cap V(\overline{H})| < 2n+5 + n+2 = 3n+7$, a contradiction.

\medskip
\noindent {\bf Claim 2.2} {\it $\overline{C} \cap (\overline{X} \cup \overline{Y}) = \emptyset$.}
\medskip

\noindent {\it Proof of Claim 2.2} 
We know by Claim 2.1 that $\overline{C} \cap (V(\overline{G}) \cup \overline{Z}) = \emptyset$. By contradiction, suppose that $\overline{C}$ contains a vertex from $\overline{X}$ or $\overline{Y}$. 
Let $\overline{A} = N_{\overline{H}}(\overline{C})$. In this case, we have that $V(\overline{G}) \cup \overline{Z} \subseteq \overline{A}$. Hence, $| \overline{A} | \geq n+2$. It follows from Proposition~\ref{prop:vizinhancaConjConvexo} that $|S \cap V(H)| = |V(H)| - |A|  \leq 6n+4 - (n+2)  = 5n+2$. 

Since $\overline{C} \cap (V(\overline{G}) \cup \overline{Z}) = \emptyset$ and $\overline{U} \cup \overline{X}$ is an independent set, the maximum order of a clique in $V(\overline{H})$ is $|\overline{Y}| = 2$; then $|S \cup V(\overline{H}) | \leq 2$. This implies that $|S| = |S \cap V(H)|+|S \cap V(\overline{H})| \leq 5n+2 + 2 = 5n+4$. Since $k \geq 3$, we have that $|S| \geq k+5n+3 = 5n+6$, a contradiction.
\hfill\ensuremath{\square}

\smallskip

By Claims 2.1 and 2.2, we have that a proper convex set $S$ of $H\overline{H}$ of order at least $k + 5n + 3$ is such that $S \cap \overline{U} \neq \emptyset$ or $S \cap V(\overline{H}) = \emptyset$.
We show that, in both cases, the order of $S$ implies in $|S \cap V(G)| \geq k$.

\medskip
\noindent {\bf Case 3.1} {\it $S \cap \overline{U} \neq \emptyset$.}
\medskip

Since $\overline{U}$ is an independent set, $|S \cap \overline{U}| \leq 1$. So, let $i \in [n]$, and consider that $\overline{u}_i \in S$. We have that $\overline{A} = N_{\overline{H}}(\overline{u}_i) = \overline{Z}$. By Proposition~\ref{prop:vizinhancaConjConvexo}, $S \cap Z = \emptyset$. 
By the construction of $H$, we have that $|U \cup X \cup Y| = 5n+2$.
Since $|S \cap V(H)| = |S| - 1 = k+5n+2$, we have that $|S \cap V(G)| \geq k$.

\medskip
\noindent {\bf Case 3.2} {\it $S \cap V(\overline{H}) = \emptyset$.}
\medskip

By Claim 1, we have that 

\begin{center}
\begin{tabular}{ccl}
either & $S \cap V(H) \subseteq (U \cup V(G) \cup X \cup Y)$ & (I) \\ 
or & $S \cap V(H) \subseteq (U \cup V(G) \cup Z)$ & (II). \\ 
\end{tabular} 
\end{center}

Condition II implies that the order of $S$ is at most $2n+2$, a contradiction. Then, consider that $S \subseteq (U \cup V(G) \cup X \cup Y)$. Still by the construction of $H$, $|U \cup X \cup Y| = 5n+2$. Since $|S \cap V(H)| = |S| = k+5n+3$, we obtain that $|S \cap V(G)| \geq k+1$.

\medskip
By Cases 3.1 and 3.2 we have that $|S \cap V(G)| \geq k$. It remains to show that $S \cap V(G)$ is a clique. Suppose, by contradiction, that $S \cap V(G)$ is not a clique. Then, there exist $v_1, v_2 \in V(G)$ such that $v_1v_2 \notin E(G)$. This implies that $Z \subseteq I_{H\overline{H}}[v_1,v_2]$. But, in both cases, $S \cap \overline{U} \neq \emptyset$ and  $S \cap V(\overline{H}) = \emptyset$; thus we have that $S \cap Z = \emptyset$, a contradiction. Therefore $S \cap V(G)$ is a clique of order at least $k$, which completes the proof. 
\end{proof}

Figure~\ref{fig:construcaoConGrafoVizinhancas} contains an example of graph $H\overline{H}$ constructed for Theorem~\ref{theo:conNPComplete}. Every edge joining two rectangles $A$ and $B$ represents the set of all edges joining every pair of vertices $a \in A$ and $b \in B$. The black vertices correspond to the convex set $S$ except the vertices from $V(G)$. For convenience, the edges joining corresponding vertices from $G$ to $\overline{G}$ are not depicted in the figure. 

\begin{figure}[htb]
\centering
{\setlength{\fboxsep}{10pt}
\setlength{\fboxrule}{0.3pt}
\fbox{
\includegraphics[width=0.85\textwidth]{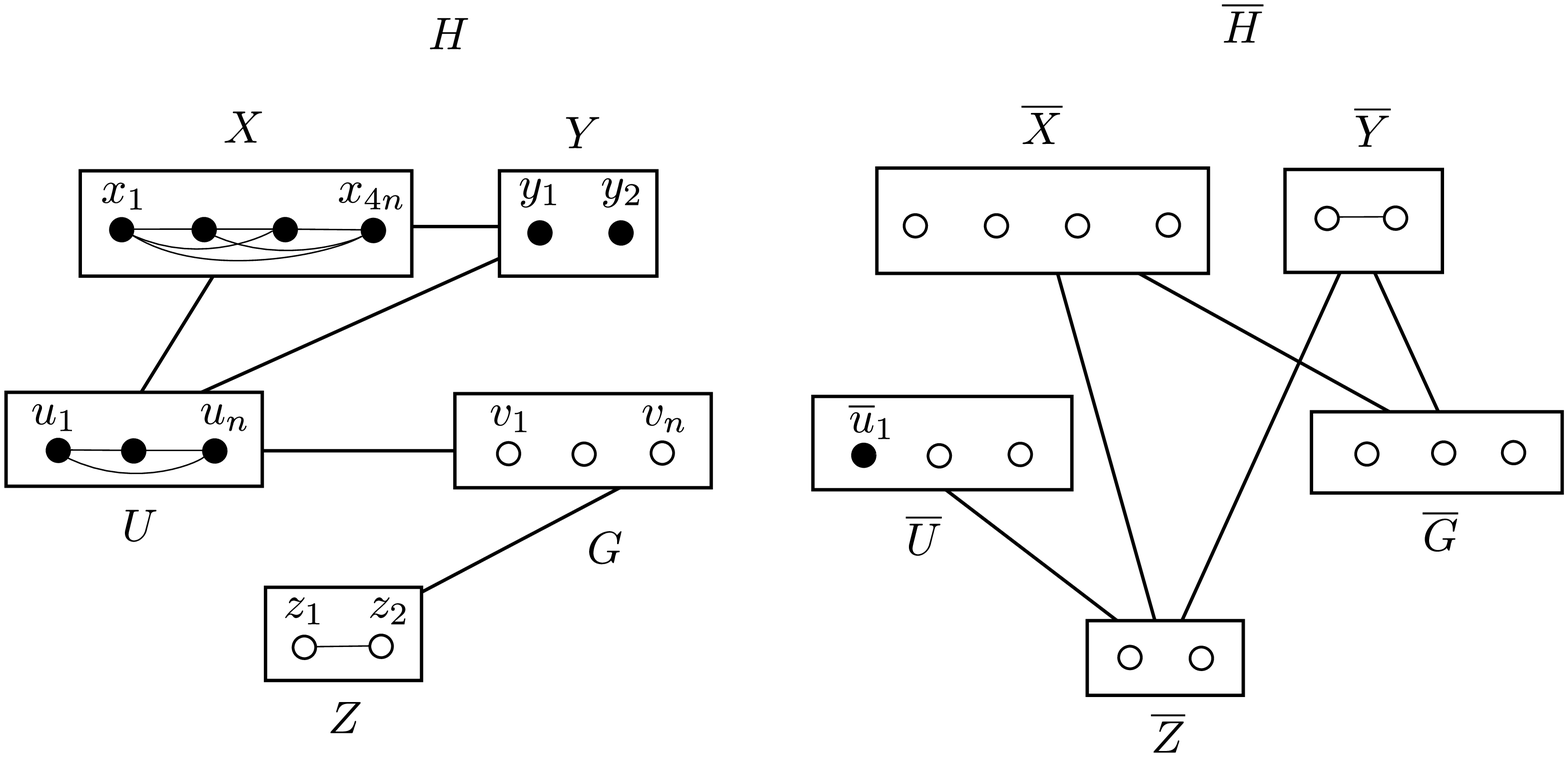}
}}
\caption{Graph $H\overline{H}$ constructed for the reduction of Theorem~\ref{theo:conNPComplete}. }
\label{fig:construcaoConGrafoVizinhancas}
\end{figure}

Notice that the graph $H$ constructed for Theorem~\ref{theo:conNPComplete} has diameter $3$. In view of the complexity result of that theorem, we finish by showing a lower bound of this parameter for graphs with restricted diameter. 

\begin{theorem}\label{theo:diam2}
Let $G$ be a graph of order $n$. If $diam(G) \neq 3$, then $con(G\overline{G}) \geq n$.
\end{theorem}

\begin{proof}
Consider first that $diam(G) \leq 2$.
Let $u,v \in V(G)$. Since $diam(G) \leq 2$, and every $(u,v)$-path passing through $V(\overline{G})$ has length at least $3$, we obtain that $I_{G\overline{G}}[u,v] \cap V(\overline{G}) = \emptyset$. Hence, $V(G)$ is a convex set of $G\overline{G}$. Since $V(G) \subset V(G\overline{G})$, it follows that $con(G\overline{G}) \geq |V(G)| = n$.

Now, let $diam(G) \geq 4$. According to  \cite{goddard2011distance}, $diam(G) > 3$ implies that $diam(\overline{G}) \leq 2$. Since $G\overline{G}$ is isomorphic to $\overline{G}G$, the result follows from the above case. 
\end{proof}

\section{Conclusions} 
\label{sec:conclusions}

We have considered the convexity number in the geodetic convexity for complementary prisms $G\overline{G}$. When $G$ is disconnected or $G$ is a cograph we provided an equality. When $diam(G) \neq 3$ we have presented a lower bound. From the complexity point of view, we have proved that, given a complementary prism $H\overline{H}$, and an integer $k$, it is $\NP$-complete to decide whether $con(H\overline{H}) \geq k$. 

\acknowledgements
\label{sec:ack}
The authors would like to thank CAPES and FAPEG for the partial support.

\bibliographystyle{abbrvnat}
\bibliography{bibliografia_convexidade}
\label{sec:biblio}

\end{document}